\documentclass[10pt,reqno]{amsart}
\usepackage{amssymb,mathrsfs,color}
\usepackage{pinlabel}

\usepackage{amsmath, amsfonts, amsthm, verbatim, amssymb}
\usepackage{epstopdf, mathtools}
\usepackage{color}
\usepackage{esint}
\usepackage{stmaryrd}
\usepackage{graphicx}
\usepackage{xcolor} 
\usepackage{tensor}
\usepackage{slashed}
\usepackage{cite}
\mathtoolsset{showonlyrefs}
\usepackage{fullpage}
\allowdisplaybreaks[3]
\newcommand{\bs}[1]{\boldsymbol{#1}}

\newcommand{\cl}[1]{\mathcal{#1}}

\newcommand{\al}{\alpha}

\newcommand{\eps}{\epsilon}

\newcommand{\la}{\lambda}

\newcommand{\p}{\partial}

\usepackage[multiple]{footmisc}
\numberwithin{equation}{section}

\newtheorem{thm}{Theorem}[section]

\newtheorem{prop}[thm]{Proposition}

\theoremstyle{remark}

\newcommand{\tr}{\mathrm{tr}\,}

\renewcommand{\div}{\mathrm{div}}
\newcommand{\rar}{\rightarrow}

\newcommand{\bbC}{\mathbb C}

\newcommand{\bbG}{\mathbb G}

\newcommand{\bbK}{\mathbb K}

\newcommand{\bbR}{\mathbb R}

\newcommand{\tens}{\otimes}

\newcommand{\ReG}{\mathrm{Re}\,G}
\newcommand{\ImG}{\mathrm{Im}\,G}
\newcommand{\ReQ}{\mathrm{Re}\,Q}
\newcommand{\ImQ}{\mathrm{Im}\,Q}

\begin{document}

\title[Surface strain-gradient elastic crack faces]{Mixed-mode loading of a straight crack with surface strain-gradient elasticity}
\dedicatory{Dedicated to D. J. Steigmann, in recognition of his profound influence on the \\ interplay between Geometry and Mechanics.}
\author{C. Rodriguez, A. Zemlyanova}

\begin{abstract}
This work models brittle fracture using a linearized surface-substrate theory in which the crack faces possess surface stresses derived from a surface strain-gradient elastic energy. The model incorporates surface stretching, curvature, and surface gradients of stretching into the surface energy, thereby capturing small-length-scale effects absent from earlier surface elasticity formulations. The theory, supplemented with physically motivated tip conditions, is applied to the mixed mode-I/mode-II loading of a finite straight crack in an infinite isotropic plate. Using complex-analytic techniques, it is shown that the resulting stress and strain fields remain bounded up to the crack tips for nearly all admissible parameter values. Combined with previous results for mode-III loading, the analysis demonstrates that linearized surface-substrate models incorporating surface strain-gradient elasticity eliminate crack-tip singularities across all principal modes of far-field loading. 
\end{abstract}

\maketitle

\section{Introduction}

\subsection{Surface-substrate interactions and fracture modeling}

It is now well established that the interactions between a solid and its surroundings through surface forces such as surface tension play a decisive role in determining mechanical behavior at small length scales, particularly near boundaries. This recognition, already implicit in the classical studies of fracture by Griffith, Westergaard, and Irwin \cite{Griffith1921, Westergaard1939, Irwin1956, Irwin1957, Anderson2005}, motivated efforts to represent surface effects within continuum theory. A natural mathematical framework for doing so is to endow portions of a body’s boundary with their own thermodynamic structure, characterized by an energy density distinct from that of the bulk. This idea forms the basis of modern theories of material surfaces, which extend classical continuum mechanics by assigning energetic and constitutive properties to both boundary and interior (interfacial) surfaces.

The first comprehensive formulation of such a theory was given by Gurtin and Murdoch \cite{GurtinMurd75, GurtinMurd78}, who sought to model compressive surface stresses observed in cleaved crystals. Their framework treats the surface as a membrane that resists stretching but not bending and has become a cornerstone in the study of small-scale phenomena confined to thin boundary or interfacial regions. Although the Gurtin-Murdoch theory has been applied successfully to a wide range of mechanical problems, Steigmann and Ogden \cite{SteigOgden97a, SteigOgden99} later showed that it cannot support compressive energy-minimizing equilibrium states.

This limitation was overcome by Steigmann and Ogden \cite{SteigOgden97a, SteigOgden99}, who extended the surface energy to include curvature dependence, thereby endowing the material surface with flexural stiffness. The resulting Steigmann-Ogden model has since been employed in the analysis of contact \cite{Zem18, MI2018, Zem19, LiMi19a, LiMi19b, LiMi21, ZemWhite22}, inclusion \cite{HAN2018, ZEMLYANOVA2018, ZEMMOG18B, DAI2018, MogKZem19, WANG2020311}, and brittle fracture problems \cite{Zem17StraightCrack, Zem21Penny, Zem20Multiple}, thereby establishing a unified theoretical framework for surfaces endowed with both in-plane and flexural stiffness. The success of these fracture-related applications provides the principal motivation for the present work.

Classical linear elastic fracture mechanics (LEFM) remains the most prominent framework for brittle fracture modeling, even though \textit{it predicts unbounded stresses and strains at crack tips}. Derived under the assumption of infinitesimal strains, the governing equations of LEFM produce singular solutions that are physically inconsistent. Numerous extensions of the theory have been proposed to regularize these singularities (see, e.g., \cite{Broberg}). Theories incorporating higher gradients of the displacement field in the bulk's stored energy beginning with the couple-stress formulations of Muki and Sternberg \cite{SternbergMuki65, SternbergMuki67} and Bogy and Sternberg \cite{BogySternberg67, BogySternberg68} eliminated singularities in the rotation gradient but left stress singularities intact. Subsequent strain-gradient theories \cite{AltanAif1992, Askes2011} succeeded in regularizing strains but not stresses. Collectively, these developments revealed that including higher-order gradients in the bulk's stored energy can moderate, but not fully remove, the singularities inherent in classical LEFM.

A different strategy for resolving crack-tip singularities emerged through the modification of boundary conditions rather than bulk constitutive laws. In this approach, the crack faces are modeled as material surfaces endowed with their own surface stresses, introducing higher surface gradients without altering the bulk's standard isotropic stored energy. Beginning with \cite{OhWaltonSlatt} and refined by Sendova and Walton \cite{SendovaWalton}, curvature-dependent surface tensions were postulated that could eliminate singularities for a number of problems \cite{ZemWalton12, Zem12Curvature, Zem14, Walton2014plane, Ferguson2015numerical}. However, these formulations lack a variational basis: the resulting surface stresses cannot, in general, be derived from an energy density. When surface stresses are instead derived from a Gurtin-Murdoch energy, singularities persist under far-field loading \cite{WaltonNote12, IMRUSch13}. More refined analyses \cite{Gorbushinetal20} indicate that inhomogeneous surface properties may alter the intensity of the singularities, but they do not remove them. By contrast, Steigmann-Ogden surface elasticity eliminates singularities for certain plane strain and axisymmetric cracks \cite{Zem17StraightCrack, Zem20Multiple, Zem21Penny}, though for anti-plane shear (mode-III loading), the curvature dependence vanishes and the crack-tip fields remain unbounded. The persistent failure of existing surface elasticity models under mode-III loading motivated the development of a theory that predicts bounded stresses and strains under all primary loading modes.

An extension in this direction was recently proposed in \cite{rodriguez2024elastic}, where the Steigmann-Ogden surface energy was augmented to include the surface gradient of stretching, producing an energy of surface strain-gradient type:
\begin{align}
	E = \int_{\mathcal B} W(\boldsymbol E), dA + \int_{\mathcal S} U(\boldsymbol E_s, \boldsymbol\rho, \nabla_s \boldsymbol E_s), dA. \label{eq:introenergy}
\end{align}
where $\cl B$ is the body, $\mathcal S \subseteq \p \cl B$ is the material surface with surface stresses, $\bs E$ is the Lagrange strain, $\bs E_s$ is the surface strain tensor, $\bs \rho$ is the relative normal curvature tensor, and $\nabla_s$ is the surface gradient.\footnote{See Section 2 for the precise definitions of these kinematic descriptors.} For planar surfaces $\mathcal S$, the functional form of $U$ coincides with that of Hilgers and Pipkin’s strain-gradient plate theory \cite{HilgPip92a}, while for curved surfaces it agrees with Steigmann’s theory of stand-alone strain-gradient elastic shells \cite{Steigmann18Lattice}. Physically, the inclusion of $\nabla_s \boldsymbol E_s$ endows the surface with resistance to geodesic distortion, namely, the tendency of convected geodesics to depart from geodesic paths on the deformed surface \cite{rodriguez2024elastic}.

In \cite{rodriguez2024elastic, rodriguez2024midsurface}, it was shown that the linearized surface-substrate model derived from \eqref{eq:introenergy} via the principle of virtual work eliminates crack-tip singularities for far-field mode-III loading in finite-length cracks, in contrast with the Gurtin-Murdoch and Steigmann-Ogden models. The work \cite{rodriguez2024elastic} employed an ad hoc quadratic surface energy $U$ suggested by Hilgers and Pipkin \cite{HilgPip92a}, and the subsequent work \cite{rodriguez2024midsurface} used the surface energy
\begin{gather}
	U = \frac{\la_s}{2} (\tr\boldsymbol E_s)^2 + \mu_s|\boldsymbol E_s|^2
	+ \ell_s^2\left(\frac{\la_s}{2} (\nabla_s \tr\boldsymbol E_s)^2 + \mu_s|\nabla_s \boldsymbol E_s|^2\right)
	+ \frac{\zeta}{2}(\tr\boldsymbol\rho)^2 + \eta|\boldsymbol\rho|^2, \label{eq:surfaceenergy}
\end{gather}
where $\ell_s$ a material length scale.
In contrast to the phenomenological energy used in \cite{rodriguez2024elastic}, the form \eqref{eq:surfaceenergy} can be derived by dimensional reduction from the three-dimensional stored energy
\begin{align}
	W_{sg} = \frac{\lambda}{2}(\tr\boldsymbol E)^2 + \mu|\boldsymbol E|^2 + \ell_s^2\left[\frac{\lambda}{2}(\nabla \tr\boldsymbol E)^2 + \mu|\nabla \boldsymbol E|^2\right], \label{eq:aifantis}
\end{align}
for a flat strain-gradient elastic layer of thickness $h$, resulting in the identifications
\begin{align}
	\la_s = \frac{2h\la \mu}{\la + 2 \mu}, \quad \mu_s = h \mu, \quad \zeta = \Bigl (\frac{h^3}{24} + h \ell_s^2 \Bigr ) \la_s, \quad \eta = \Bigl (\frac{h^3}{24} + h \ell_s^2 \Bigr ) \mu_s.
\end{align}
The analysis in \cite{Rodriguezetal25} obtained analogous surface energies from the three-dimensional perspective through similar asymptotic arguments for shells and for Toupin–Mindlin energies \cite{Toupin64, Mindlin64a} more general than \eqref{eq:aifantis}. Strong ellipticity of the surface energies used in \cite{rodriguez2024elastic} and \eqref{eq:surfaceenergy} played a crucial role in guaranteeing the boundedness of the stresses and strains predicted by the theory, in contrast to the $\ell_s=0$ limit corresponding to the Steigmann-Ogden model. 

\subsection{Main results and outline}

The present work applies the linearized surface-substrate model of \cite{rodriguez2024elastic} with surface energy \eqref{eq:surfaceenergy} to the mixed mode-I/mode-II loading of a finite straight crack in an infinite isotropic plate. Using complex-analytic methods, we demonstrate that the stresses and strains remain bounded up to the crack tips for nearly all admissible parameter values. Together with the results for mode-III loading \cite{rodriguez2024elastic, rodriguez2024midsurface}, this establishes that linearized surface-substrate models of fracture with surface strain-gradient elasticity regularizes all principal modes of far-field fracture loading. Thus, the resulting theory provides a unified physically consistent framework for modeling brittle fracture with bounded stress and strain fields.

Section 2 formulates the linearized surface-substrate model, defining the stored energies and deriving the governing equations for far field, mixed-mode, plane strain crack loading via a principle of virtual work. Particular attention is given to the tip conditions, which play a central role in ensuring bounded stresses and strains up to the crack termini. These include the cusp-closing conditions, the continuity of stresses across the crack tips, and the vanishing of resultant surface couples, a physically reasonable requirement for both plane strain and antiplane shear configurations. Section 3 reformulates the mixed-mode boundary-value problem using complex-analytic potentials, leading to a system of coupled singular integro-differential equations. Section 4 reduces this system to a Fredholm-type integral formulation and establishes existence, uniqueness, and sufficient regularity of solutions to guarantee the boundedness of the stress and strain fields up to the crack tips (modulo a countable set of values of the model's parameters). Section 5 presents numerical results based on Chebyshev polynomial approximations and compares the outcomes across parameter regimes, illustrating the quantitative influence of the surface strain-gradient contribution to the energy on the model.

\subsection*{Acknowledgements} C. R. gratefully acknowledges support of NSF DMS-2307562. A. Z. gratefully acknowledges support of Simons Foundation, award numbers 713080 and SFI-MPS-TSM-00013162.

\section{Preliminaries}

In this section, we introduce the linearized surface-substrate model used throughout the paper. We present the substrate and surface stored energies and their infinitesimal displacement approximations. The governing field equations for the mixed mode-I/mode-II loading of a finite straight crack are obtained via a principle of virtual work, and the tip conditions needed to guarantee bounded stress and strain fields are stated explicitly and justified on physical grounds.

\subsection{Green elastic bodies with strain-gradient elastic boundary surfaces} We now briefly review the foundations of the surface-substrate theory from \cite{rodriguez2024elastic} for equilibrium states of a Green elastic body $\mathcal B \subset \mathbb R^3$ with a surface strain-gradient elastic boundary surface $\mathcal S \subseteq \p \mathcal B$.\footnote{In this work, we identify three-dimensional Euclidean space with $\bbR^3$ via a choice of origin and a fixed orthonormal basis $\{\bs e_1, \bs e_2, \bs e_3\}$.}  

Let $\boldsymbol x: U \rar \mathcal S$ be a local parameterization of $\mathcal S$. Here $U \subseteq \mathbb R^2$ is an open set with coordinates $(\theta^\al)$. We denote the corresponding natural basis vectors by $\boldsymbol A_\al := \boldsymbol x_{,\al}$ and the dual basis vectors by $\boldsymbol A^\al$ respectively. Here $\mbox{}_{,\al} = \p_{\theta^\al}$. The projection onto the tangent plane of $\mathcal S$ is denoted by $\bs 1$ and given by
\begin{align}
	\bs 1 = \bs A_\alpha \otimes \bs A^\alpha. 
\end{align}
The components of the metric are given by 
\begin{align}
	A_{\al \beta} := \boldsymbol A_{\al} \cdot \boldsymbol A_{\beta},
\end{align}  
the determinant of $(A^{\al \beta})$ is denoted by $A$, and the dual components are denoted by $(A^{\al \beta})$, so $A_{\al \gamma}A^{\gamma \beta} = \delta_\al^\beta$. The components of tensors are raised and lowered using $(A^{\al\beta})$ and $(A_{\al \beta})$ respectively. For a vector field $\boldsymbol u$ on $\mathcal S$ we define  
\begin{align}
\boldsymbol u_{\al|\beta} := \boldsymbol u_{,\al\beta} - \Gamma^{\gamma}_{\al \beta} \boldsymbol u_{,\gamma},
\end{align} 
where $(\Gamma^\gamma_{\beta\al})$ are the Christoffel symbols associated to $(A_{\al\beta})$. For a scalar valued function $\varphi$ on $\mathcal S$, its surface gradient is defined via $$\nabla_s \varphi := \varphi_{,\al} \boldsymbol A^\al,$$ and for a vector or tensor field $\boldsymbol B$ on $\mathcal S$, its surface gradient is defined via  
\begin{align}
	\nabla_s \boldsymbol B := \boldsymbol B_{,\al} \tens \boldsymbol A^\al. 
\end{align}
The surface divergence $\div_s$ and surface Laplacian $\Delta_s$ are then defined in terms of standard combinations of traces and surface gradients. 

Let $\boldsymbol \chi: \mathcal B \rar \mathbb R^3$ be a smooth invertible deformation with deformation gradient $\boldsymbol F$, inducing a smooth embedding $\boldsymbol y : \mathcal S \rar \mathbb R^3$. The surface strain tensor is defined via
\begin{align}
	\boldsymbol E_s := \frac{1}{2}(\boldsymbol y_{,\al} \cdot \boldsymbol y_{,\beta} - A_{\al \beta}) \boldsymbol A^\al \tens \boldsymbol A^\beta, 
\end{align} 
and the relative normal curvature tensor is given by 
\begin{align}
	\boldsymbol \rho := (\boldsymbol n \cdot \boldsymbol y_{,\al \beta} - \boldsymbol N \cdot \boldsymbol x_{,\al\beta}) \boldsymbol A^\al \tens \boldsymbol A^\beta, 
\end{align}
where $\boldsymbol n$ is the outward unit normal vector field on the convected surface $\boldsymbol \chi(\mathcal S)$ and $\boldsymbol N$ is the outward unit normal vector field on $\mathcal S$. This induces an orientation of $\mathcal S$ by requiring that $\bs N = |\bs A_1 \times \bs A_2|^{-1} \bs A_1 \times \bs A_2$. 

For the associated displacement vector field $\boldsymbol u(\boldsymbol p) := \boldsymbol \chi(\boldsymbol p) - \boldsymbol p$ on $\mathcal B$, we define the infinitesimal strain, surface strain, and relative normal curvature tensors via 
\begin{align}
		\bs \eps = \frac{1}{2}(\nabla \bs u + (\nabla \bs u)^T), \quad \boldsymbol \eps_s = \frac{1}{2}\bigl (\boldsymbol A_\al \cdot \boldsymbol u_{,\beta} + \boldsymbol u_{,\al} \cdot \boldsymbol A_{,\beta} \bigr ) \boldsymbol A^\al \tens \boldsymbol A^\beta, \quad 
	\boldsymbol \kappa_s = \boldsymbol N \cdot \boldsymbol u_{,\al|\beta} \boldsymbol A^\al \tens \boldsymbol A^\beta. 
\end{align}

We assume that the substrate $\mathcal B$ is modeled by the small-strain quadratic isotropic stored energy
\begin{align}
	W = \frac{\la}{2}(\mathrm{tr}\, \boldsymbol E)^2 + \mu|\boldsymbol E|^2,
\end{align}
and that the surface $\mathcal S$ is endowed with the strain-gradient surface energy incorporating surface tension:
\begin{align}
	U &= \sigma_0 J_{\bs 1 + 2 \bs E_s}^{1/2} + \frac{\la_s}{2}(\mathrm{tr}\, \boldsymbol E_s)^2 + \mu_s |\boldsymbol E_s|^2 + \ell_s^2 \Bigl [ \frac{\la_s}{2}|\nabla_s \mathrm{tr}\, \boldsymbol E_s|^2 + \mu_s |\nabla_s \boldsymbol E_s|^2 \Bigr ]
	+ \frac{\zeta}{2}(\mathrm{tr}\, \boldsymbol \rho)^2 + \eta |\boldsymbol \rho|^2. \label{eq:nonlinsurf}
\end{align}
Here, $\sigma_0\in\mathbb R$ denotes the surface tension constant and $J_{\bs a}$ is the Jacobian of a tensor field $\bs a$ on $\mathcal S$. 

If we assume that the displacement satisfies 
\begin{align}
	\max_{\cl B} |\nabla \bs u| + L \max_{\mathcal S} |\nabla_s \nabla_s \bs u| := \delta \ll 1, \label{eq:infdis}
\end{align}
where is $L$ is a fixed length scale, it follows that, up to $o(\delta^2)$ terms,
\begin{align}
	W &= \frac{\la}{2}(\mathrm{tr}\, \boldsymbol \eps)^2 + \mu|\boldsymbol \eps|^2, \label{eq:infbulk} \\
\begin{split}
	U &= \sigma_0(1 + \div_s \bs \eps_s) + \frac{\la_s+\sigma_0}{2}(\mathrm{tr}\, \boldsymbol \eps_s)^2 + (\mu-\sigma_0) |\boldsymbol \eps_s|^2 + \frac{\sigma_0}{2}|\nabla_s \bs u|^2 
	\\ &+\ell_s^2 \Bigl [ \frac{\la_s}{2}|\nabla_s \mathrm{tr}\, \boldsymbol \eps_s|^2 + \mu_s |\nabla_s \boldsymbol \eps_s|^2 \Bigr ] + \frac{\zeta}{2}(\mathrm{tr}\, \boldsymbol \kappa_s)^2 + \eta |\boldsymbol \kappa_s|^2.
\end{split} \label{eq:infsurf}
\end{align}
When $\ell_s = 0$ and $\zeta = \eta = 0$, \eqref{eq:infsurf} reduces to the classical Gurtin-Murdoch surface energy, and when only $\ell_s=0$, one recovers the infinitesimal Steigmann-Ogden model, which includes resistance to flexure but no surface strain-gradient terms. In this study, we assume that 
\begin{align}
\zeta + 2 \eta > 0, \quad \lambda_s + 2 \mu_s > 0, \quad \ell_s > 0. 
\end{align}

In anticipation of the plane strain setting considered in this work, we specialize the surface energy \eqref{eq:infsurf} to planar deformations. For plane strain displacements of the form
\[
\boldsymbol u = u_1(x_1,x_2)\boldsymbol e_1 + u_2(x_1,x_2)\boldsymbol e_2
\]
and for a material surface $\mathcal S$ lying in the $x_1$-$x_3$ plane, a direct calculation shows that \eqref{eq:infsurf} reduces to
\begin{align}
	U = \sigma_0 (1 \pm u_{1,1}) 
	+ \frac{\la_s+2\mu_s}{2} \Bigl [ (u_{1,1})^2 + \ell_s^2 (u_{1,11})^2 \Bigr ]
	+ \frac{\sigma_0}{2} (u_{2,1})^2 + \frac{\zeta + 2\eta}{2} (u_{2,11})^2, \label{eq:simpsurfen}
\end{align}
where the choice of sign corresponds to the orientation of the surface normal, $\boldsymbol N = \mp \boldsymbol e_2$.

\subsection{Formulation of the mixed-mode problem}
In this work, we study a finite crack of length $2\ell$ lying along the interval $[-\ell, \ell]$ in an infinite brittle\footnote{The adjective ``brittle" indicates that we are only considering infinitesimal displacement gradients and the linearized theory derived under the assumption \eqref{eq:infdis}.} plate $\mathcal B = \mathbb{R}^2 \times [-H,H]$, under far-field mixed-mode loading conditions in plane strain
\begin{align}
	\boldsymbol u = u_1(x_1, x_2)\boldsymbol e_1 + u_2(x_1, x_2) \boldsymbol e_2.
\end{align}
The far-field loading conditions are expressed via 
\begin{align}
	\lim_{|(x_1,x_2)| \rar \infty} \boldsymbol \sigma(\boldsymbol x) = \boldsymbol \sigma^\infty
\end{align}
where 
\begin{align}
	\bs \sigma = \la (\tr \bs \eps) \bs I + 2\mu \bs \eps,
\end{align}
is the Cauchy stress tensor. 
The two crack faces, approached as $y \to 0^\pm$, are denoted by $\mathcal S_{\pm} := [-\ell, \ell] \times \{0^\pm\} \times [-H,H]$. Their union is endowed with the surface stored energy \eqref{eq:infsurf}; see \eqref{eq:simpsurfen} for plane strain. For simplicity, we assume that no external tractions are applied on the crack surfaces. The endpoints of the crack, $\{\pm \ell\} \times \{0\} \times [-H,H]$, will be referred to as the \textit{crack tips}, consistent with the two-dimensional reduction from the plane strain assumption.

We now derive the governing field equations via a principle of virtual work.
\begin{prop}\label{p:1}
Let $\mathcal P \subset \bbR^2$ be an open set with piecewise smooth boundary and $[-\ell,\ell] \subset \mathcal P$. We say that $\boldsymbol u: \mathcal P \rar \bbR^2$ is kinematically admissible on $\mathcal P$ and write $\boldsymbol u \in \mathcal A(\mathcal P)$ if: 
\begin{enumerate}
	\item $\boldsymbol u \in C^2(\mathcal P \backslash [-\ell,\ell])$. 
	\item $\boldsymbol u|_{\mathcal P \cap \{\pm x_2 > 0\}}$ extend to $C^2$ functions $\boldsymbol u^\pm$ on $\mathcal P \cap \{\pm x_2 \geq 0 \}$ and $\boldsymbol u^{\pm} |_{[-\ell,\ell]} \in C^4([-\ell,\ell]).$
	\item $\boldsymbol u^\pm$ satisfy the cusp conditions
\begin{align}
	\boldsymbol u^{+}(\pm \ell, 0) &= \boldsymbol u^-(\pm \ell,0), \label{eq:6} \\
	\boldsymbol u^{+}_{,1}(\pm \ell,0) &= \boldsymbol u_{,1}^+(\pm \ell,0).   \label{eq:7}
\end{align}  
\end{enumerate}
 For $\boldsymbol u \in \mathcal A(\mathcal P)$, define the total energy $E(\cdot; \mathcal P)$, edge forces $\boldsymbol t_e$ (at $t = \pm \ell)$, and edge double forces $\boldsymbol m_e$ (at $t = \pm \ell$) via 
\begin{align}
	E(\boldsymbol u; \mathcal P) 
	&:= \int_{\mathcal P \backslash [-\ell,\ell]} \Bigr [ \frac{\la}{2} (\tr \boldsymbol \eps)^2 + \mu|\boldsymbol \eps|^2 \Bigr ] dA + \sigma_0 \sum_{\pm} \int_{-\ell}^\ell(1 \pm u_{1,1}^\pm) dt \\ &+ \sum_{\pm } \int_{-\ell}^\ell \frac{\la_s+2\mu_s}{2} \Bigl [ (u_{1,1}^\pm )^2 + \ell_s^2 (u_{1,11}^\pm)^2 \Bigr ] dt
	+ \sum_{\pm } \int_{-\ell}^\ell \Bigl [\frac{\sigma_0}{2} (u_{2,1}^\pm)^2 + \frac{\zeta + 2\eta}{2} (u_{2,11}^\pm)^2 \Bigr ] dt, \\
	\boldsymbol t_e &:= \sum_{\pm} \Bigl [(\la_s + 2\mu_s)(u_{1,1}^\pm - \ell_s^2 u_{1,111}^\pm)\boldsymbol e_1 + (\sigma_0 u_{2,1}^\pm - (\zeta + 2\eta) u_{2,111}^\pm ) \boldsymbol e_2\Bigr ], \label{eq:edgef}\\
	\boldsymbol m_e &:=  \sum_{\pm} \Bigl [ \ell_s^2(\la_s + 2\mu_s) u_{1,11}^\pm \boldsymbol e_1 + (\zeta + 2\eta) u_{2,11}^\pm \boldsymbol e_2 \Bigr ]. \label{eq:edgedoublef}
\end{align}
Then $\boldsymbol u \in \mathcal A(\bbR^2)$ satisfies the principle of virtual work that for all bounded $\mathcal P$, for all $\boldsymbol v \in \mathcal A(\cl P)$,
\begin{align}
	\frac{d}{d\eps}E(\boldsymbol u + \eps \boldsymbol v; \mathcal P) \Big |_{\eps = 0} &= \int_{\p \mathcal P} \boldsymbol \sigma\boldsymbol N \, \cdot \boldsymbol v \,ds + \boldsymbol t_e \cdot \boldsymbol v \big |_{t = -\ell}^{t=\ell} + \boldsymbol m_e \cdot \boldsymbol v_{,1} \big |_{t = -\ell}^{t = \ell}, \label{eq:p11}
\end{align}
where $\boldsymbol N$ is the outward unit normal vector field on $\p \mathcal P$, if and only if the following governing field equations are satisfied:
\begin{align}
	\div \, \boldsymbol \sigma &= \boldsymbol 0, \quad \mbox{on } \mathbb R^2 \backslash [-\ell,\ell], \label{eq:3} \\
		\sigma_{12}^\pm(t,0) &= \mp(\la_s + 2\mu_s)\bigl (u_{1,11}^{\pm}(t,0) - \ell_s^2 u_{1,1111}^{\pm}(t,0) \bigr ), \label{eq:4}\\
		\sigma_{22}^\pm(t,0) &= \pm(\zeta + 2\eta) u_{2,1111}^{\pm}(t,0)\mp \sigma_0 u_{2,11}^{\pm}(t,0), \quad t \in [-\ell,\ell], \label{eq:5} 
\end{align}
where, as before, $\pm$ denotes the limiting values approaching $[-\ell,\ell]$ as $y \rar 0^{\pm}$. 
\end{prop}
\begin{proof}
We first note that continuity condition \eqref{eq:6} and the fundamental theorem of calculus imply that 
\begin{align}
	\sigma_0 \sum_{\pm} \int_{-\ell}^\ell(1 \pm u_{1,1}^\pm) dt = 2 \sigma_0 \ell.
\end{align}	
Via standard integration by parts arguments and the facts that $\boldsymbol u \in \mathcal A(\mathcal P), \boldsymbol v \in \mathcal A(\mathcal P)$, we then conclude that  
\begin{align}
	\dot E &= \int_{\mathcal P\backslash[-\ell,\ell]} -\div \boldsymbol \sigma \cdot \boldsymbol v \, dA + \int_{\p \mathcal P} \boldsymbol \sigma \boldsymbol N \cdot \boldsymbol v \, ds \\ 
	&+ \sum_{\pm} \int_{-\ell}^\ell \Bigl [ \mp \boldsymbol \sigma \boldsymbol e_2 \cdot \boldsymbol v^{\pm} - (\lambda_s + 2\mu_s)(u_{1,11}^\pm - \ell_s^2 u_{1,1111}^{\pm}) v_1 - [\sigma_0 u_{,11} - (\zeta+2\eta)u_{2,1111}^\pm] v_2 \Bigr ] dx_1 \\
	&+\Bigl [ \sum_{\pm} (\la_s + 2\mu_s)(u_{1,1}^\pm - \ell_s^2 u_{1,111}^\pm) \Bigr ] v_1 \Big |_{t = -\ell}^{t=\ell} + \Bigl [\sum_{\pm} (\sigma_0 u_{2,1}^\pm - (\zeta + 2\eta) u_{2,111}^\pm ) \Bigr ] v_2 \Big |^{t = \ell}_{t = -\ell} \\ 
	&+ \Bigl [ \sum_{\pm} \ell_s^2 (\la_s + 2\mu_s) u_{1,11}^\pm \Bigr ]v_{1,1} \Big |_{t = -\ell}^{t = \ell} + \Bigl [\sum_{\pm} (\zeta + 2\eta) u_{2,11}^\pm ) \Bigr ] v_{2,1} \Big |^{t = \ell}_{t = -\ell} \\
	&= \int_{\mathcal P \backslash[-\ell,\ell]} -\div \boldsymbol \sigma \cdot \boldsymbol v \, dA + \int_{\p \mathcal P} \boldsymbol \sigma \boldsymbol N \cdot \boldsymbol v\, ds\\ 
	&+ \sum_{\pm} \int_{-\ell}^\ell \Bigl [ \mp \boldsymbol \sigma \boldsymbol e_2 \cdot \boldsymbol v^{\pm} - (\lambda_s + 2\mu_s)(u_{1,11}^\pm - \ell_s^2 u_{1,1111}^{\pm}) v_1 - [\sigma_0 u_{2,11}^{\pm} - (\zeta+2\eta)u_{2,1111}^\pm] v_2 \Bigr ] dx_1 \\
	&+ \boldsymbol t_e \cdot \boldsymbol v \big |_{t = -\ell}^{t=\ell} + \boldsymbol m_e \cdot \boldsymbol v_{,1} \big |_{t = -\ell}^\ell.
\end{align}
The result now follows from applying the fundamental theorem of the calculus of variations and the arbitrariness of $\cl P$. 
\end{proof}

We now derive balance laws from our principle of virtual work. Let $\mathcal P \subset \bbR^2$ be a bounded open set with piecewise smooth boundary and $[-\ell,\ell] \subset \mathcal P^\circ$. Let $\boldsymbol a \in \bbR^2$, $b \in \bbR$, and define 
\begin{align}
	\boldsymbol v = \boldsymbol a - b x_2 \boldsymbol e_1 + b x_1 \boldsymbol e_2 \in \mathcal A(\mathcal P).  
\end{align}
If $\boldsymbol u$ satisfies \eqref{eq:p11}, then  
\begin{align}
	\boldsymbol 0 = \Bigl [ \int_{\p \mathcal P} \boldsymbol \sigma \boldsymbol N \, ds + \boldsymbol t_e |_{t = -\ell}^{t=\ell} \Bigr ] \cdot \boldsymbol a + 
	\Bigl [
	\Bigl ( \int_{\p \mathcal P} \boldsymbol x \times \boldsymbol \sigma \boldsymbol N \, ds + \boldsymbol x \times \boldsymbol t_e \Big |_{x_1 = t = -\ell, x_2 = 0}^{x_1 = t=\ell, x_2 = 0} \Bigr ) \cdot \boldsymbol e_3 + \boldsymbol m_e \cdot \boldsymbol e_2 \big |_{t = -\ell}^{t = \ell} 
	\Bigr ]b
\end{align}
Since $\boldsymbol a$ and $b$ are arbitrary and the cross product of two vectors in the plane is always parallel to $\boldsymbol e_3$, we conclude that there hold \textit{balance of forces} for parts containing the crack
\begin{align}
	\boldsymbol 0 = \int_{\p \mathcal P} \boldsymbol \sigma \boldsymbol N \, ds + \boldsymbol t_e |_{t = -\ell}^{t=\ell}, 
\end{align}
and \textit{balance of moments} for parts containing the crack
\begin{align}
	\boldsymbol 0 = \int_{\p \mathcal P} \boldsymbol x \times \boldsymbol \sigma \boldsymbol N \, ds + \boldsymbol x \times \boldsymbol t_e \Big |_{x_1 = t = -\ell, x_2 = 0}^{x_1 = t=\ell, x_2 = 0} + (\boldsymbol e_3 \tens \boldsymbol e_2)\boldsymbol m_e \big |_{t = -\ell}^{t = \ell}. \label{eq:mombalance}
\end{align} 
The term $(\boldsymbol e_3 \tens \boldsymbol e_2)\boldsymbol m_e \big |_{t = -\ell}^{t = \ell}$ represents a couple parallel to $\boldsymbol e_3$, the direction of the crack edge.

Since the equations \eqref{eq:4} and \eqref{eq:5} involve fourth-order differential operators acting on $\boldsymbol u^{\pm}$, we require 16 linearly independent scalar conditions to ensure well-posedness.
We note that at $(x_1, x_2) = (\pm \ell,0)$, \eqref{eq:6} and \eqref{eq:7} can be represented as 
\begin{align}
	\begin{bmatrix}
		1 &-1 &0 &0 &0 &0 &0 &0 \\
		1 &-1 &0 &0 &0 &0 &0 &0 \\
		0 &0 &1 &-1 &0 &0 &0 &0 \\
		0 &0 &1 &-1 &0 &0 &0 &0 \\
		0 &0 &0 &0 &1 &-1 &0 &0 \\
		0 &0 &0 &0 &1 &-1 &0 &0 \\
		0 &0 &0 &0 &0 &0 &1 &-1 \\
		0 &0 &0 &0 &0 &0 &1 &-1
	\end{bmatrix}
	\begin{bmatrix}
		u^+_1 \\
		u^-_1 \\
		u^+_2 \\
		u^-_2 \\
		u_{1,1}^+ \\
		u_{1,1}^- \\
		u_{2,1}^+ \\
		u_{2,1}^- \\
	\end{bmatrix} = 
	\boldsymbol 0. \label{eq:boundarymatrix}
\end{align}
Thus, only eight total linearly independent scalar conditions at the two tips have been prescribed, and eight more conditions are needed. In addition, we require that the stresses be continuous at the crack tips: 
\begin{align}
	\sigma_{12}^+(\pm \ell, 0) = \sigma_{12}^-(\pm \ell,0), \quad \sigma_{22}^+(\pm \ell, 0) = \sigma_{22}^-(\pm \ell,0), \label{eq:stress}
\end{align} 
which supplies four more linearly independent scalar conditions. Finally, we require that the \textit{double forces at the crack tips vanish}: 
\begin{align}
	\boldsymbol{m}_e = \boldsymbol 0 \quad \mbox{at } t = \pm \ell. \label{eq:couple1}
\end{align}
Imposing condition \eqref{eq:couple1} ensures that the crack edge undergoes no additional twisting beyond that induced by the edge force $\boldsymbol{t}_e$ at $t = \pm \ell$ (see \eqref{eq:mombalance}) and is equivalent to
\begin{align}
	\boldsymbol u^+_{,11}(\pm \ell, 0) + \boldsymbol u^-_{,11}(\pm \ell, 0) = \boldsymbol 0. \label{eq:couple}
\end{align}

\section{The mixed-mode loading problem from the complex analytic perspective}

This section reformulates the mixed-mode crack problem using complex potentials. By expressing the fields in terms of analytic functions of a complex variable, the governing boundary-value problem is reduced to a coupled system of singular integro-differential equations for the jumps in traction and displacement derivatives across the crack. This representation isolates the analytical structure needed for application of the Fredholm alternative later.

\subsection{Savruk complex potentials} 
The stresses and derivatives of the displacements $u_1, u_2$ along $L:=[-\ell,\ell]$ can be expressed using two complex potentials \cite{muskhelishvili1953some}, $\Phi$ and $\Psi$, which results in (see, e.g., Section 1.1 of \cite{savruk2016stress})
\begin{align}
	(\sigma_{22} - i \sigma_{12})(t) &= \Phi(t) + \overline{\Phi(t)} + \frac{\overline{dt}}{dt}(t \overline{\Phi'(t)} + \overline{\Psi(t)}), \label{eq:10}\\
	2\mu \frac{d}{dt}(u_1 + i u_2)(t) &= \kappa \Phi(t) - \overline{\Phi(t)} - \frac{\overline{dt}}{dt}(t \overline{\Phi'(t)} + \overline{\Psi(t)}), \quad t \in L,   \label{eq:11}
\end{align} 
where we have omitted which side of the curve we are evaluating along. Here, $\kappa = 3-4\nu$ where $\nu$ is Poisson's ratio for the bulk material.  


The functions $\Phi$ and $\Psi$ are analytic on $\bbC \backslash L$ and can be expressed using Cauchy integrals. We adopt the form proposed by Savruk
\begin{align}
	\Phi(z) &:= \Gamma + \frac{1}{2\pi} \int_L \Bigl ( g'(t) - \frac{2i q(t)}{\kappa +1} \Bigr ) \frac{dt}{t - z}, \\
	\Psi(z) &:= \Gamma' + \frac{1}{2\pi} \int_L \Bigl ( \overline{g'(t)} - \frac{2i\kappa \overline{q(t)}}{\kappa+1}\Bigr ) \frac{\overline{dt}}{t-z} - \frac{1}{2\pi} \int_L \Bigl (g'(t) - \frac{2i q(t)}{\kappa+1} \Bigr ) \frac{\bar t dt}{(t-z)^2}
\end{align} 
where $\Gamma := \frac{1}{4}(\sigma_{11}^\infty + \sigma_{22}^\infty)$ and $\Gamma':= \frac{1}{2}(\sigma_{22}^\infty-\sigma_{11}^\infty) + i \sigma_{12}^\infty$. Inserting the Savruk representations into \eqref{eq:10} and \eqref{eq:11} yields 
\begin{align}
\begin{split}
	(\sigma_{22} - i \sigma_{12})^{\pm}(t) &= \pm q(t) + \frac{1}{2\pi} \int_L \Bigl ( \frac{2}{\tau - t} + k_1(t,\tau) \Bigr )g'(\tau) \\ 
	&+ \frac{1}{2\pi} \int_L k_2(t,\tau) \overline{g'(\tau) d\tau}
	+ \frac{1}{\pi i (\kappa+1)} \int_L \Bigl (- \frac{\kappa-1}{\tau - t} - \kappa k_1(t,\tau)) q(\tau) d\tau \\ &- \frac{1}{\pi i (\kappa+1)} \int_L k_2(t,\tau) \overline{q(\tau)d\tau}
	+ 2 \mathrm{Re}\, \Gamma + \bar \Gamma'\frac{\overline{dt}}{dt},
\end{split} \label{eq:12} \\
\begin{split}
	2\mu \frac{d}{dt}(u_1 + i u_2)^\pm(t) &= \pm\frac{i(\kappa+1)}{2} g'(t) + \frac{1}{2\pi} \int_L \Bigl (\frac{\kappa-1}{\tau - t} - k_1(t,\tau) \Bigr ) g'(\tau) d\tau \\
	&- \frac{1}{2\pi} \int_L k_2(t,\tau) \overline{g'(\tau)d\tau} + 
	\frac{\kappa}{\pi i(\kappa+1)} \int_L \Bigl ( \frac{2}{\tau - t} + k_1(t,\tau) \Bigr ) q(\tau) d\tau \\
	&+ \frac{1}{\pi i(\kappa +1)} \int_L k_2(t,\tau) \overline{q(\tau)d\tau} + \kappa \Gamma - \bar{\Gamma} - \bar\Gamma' \frac{\overline{dt}}{dt}, \quad t \in L, 
\end{split}\label{eq:13}
\end{align}
where 
\begin{align}
	k_1(t,\tau) = \frac{d}{dt}\log \frac{\tau - t}{\bar \tau - \bar t}, \quad k_2(t,\tau) = -\frac{d}{dt} \frac{\tau - t}{\bar \tau - \bar t}. 
\end{align}
Since $L = [-\ell,\ell]$, $\bar t = t$ and thus 
\begin{align}
\begin{split}
(\sigma_{22}-i\sigma_{12})^{\pm}(t) &= \pm q(t) + \frac{1}{\pi} \int_L \frac{g'(\tau)}{\tau - t}d\tau - \frac{\kappa-1}{\pi i (\kappa+1)} \int_L \frac{q(\tau)}{\tau - t}d\tau \\
&+ 2 \mathrm{Re}\, \Gamma + \bar \Gamma', 
\end{split}\label{eq:14} \\
\begin{split}
2\mu \frac{d}{dt}(u_1 + i u_2)^{\pm}(t) &= \pm \frac{i(\kappa+1)}{2}g'(t) + \frac{\kappa-1}{2\pi} \int_L \frac{g'(\tau)}{\tau - t}d\tau \\ 
&+ \frac{2\kappa}{\pi i(\kappa+1)} \int_L \frac{q(\tau)}{\tau - t}d\tau 
+ \kappa \Gamma - \bar \Gamma - \bar \Gamma', \quad t \in L.
\end{split} \label{eq:15}
\end{align}

Equations \eqref{eq:12} and \eqref{eq:13} imply that the functions $q(t)$ and $g'(t)$ represent the jumps in the stress and displacement gradient fields across the crack $L$:
\begin{align}
2q(t) &= (\sigma_{22} - i\sigma_{12})^+(t) - (\sigma_{22} - i\sigma_{12})^-(t), \label{eq:8}\\
\frac{i(\kappa+1)}{2\mu} g'(t) &= \frac{d}{dt}(u_1 + i u_2)^+(t) - \frac{d}{dt}(u_1 + i u_2)^-(t), \quad t \in L. \label{eq:9}
\end{align}
This identification provides $q$ and $g'$ with a clear physical interpretation as measures of stress and strain discontinuities along the crack faces. Moreover, integrating \eqref{eq:9} from $-\ell$ to $t$ and using \eqref{eq:6}, we see that 
\begin{align}
 \frac{i(\kappa+1)}{2\mu} g(t) &= (u_1 + i u_2)^+(t) - (u_1 + i u_2)^-(t), \quad t \in L,
\end{align}
and thus, $g(t)$ represents the jump in displacement along the crack. 

\subsection{The boundary conditions in terms of Cauchy integrals}

Adding and subtracting (with respect to $\pm$) the two boundary conditions \eqref{eq:4} and \eqref{eq:5} on the upper and lower crack faces, we obtain on $[-\ell,\ell]$,
\begin{align}
	\sigma_{12}^+ + \sigma_{12}^{-} &= 
	-(\lambda_s + 2\mu_s)(u_{1,11}^+-u_{1,11}^-) + \ell^2_s (\lambda_s + 2\mu_s)(u_{1,1111}^+-u_{1,1111}^-), \label{eq:16}\\
	\sigma_{22}^+ + \sigma_{22}^{-} &= 
	-\sigma_0(u_{2,11}^+-u_{2,11}^-) + (\zeta +2\eta) (u_{2,1111}^+-u_{2,1111}^-), \label{eq:17} \\
		\sigma_{12}^+ - \sigma_{12}^{-} &= 
	-(\lambda_s + 2\mu_s)(u_{1,11}^++u_{1,11}^-) + \ell^2_s (\lambda_s + 2\mu_s)(u_{1,1111}^++u_{1,1111}^-), \label{eq:18}\\
	\sigma_{22}^+ - \sigma_{22}^{-} &= 
-\sigma_0(u_{2,11}^++u_{2,11}^-) + (\zeta +2\eta) (u_{2,1111}^++u_{2,1111}^-). \label{eq:19}
\end{align}  
Using \eqref{eq:12} and \eqref{eq:13}, the previous conditions are equivalent to  
\begin{align}
-\mathrm{Im} \Bigl [
\frac{1}{\pi} \int_{-\ell}^\ell \frac{g'(\tau)}{\tau - t}d\tau& - \frac{\kappa-1}{\pi i (\kappa+1)} \int_{-\ell}^\ell \frac{q(\tau)}{\tau - t}d\tau + 2 \mathrm{Re}\, \Gamma + \bar \Gamma' 
\Bigr ] \\
&= \frac{1}{4\mu}(\lambda_s + 2\mu_s)(\kappa+1) \mathrm{Im}\, g''(t) - \frac{\ell_s^2}{4\mu}(\lambda_s + 2\mu_s)(\kappa+1) \mathrm{Im}\, g''''(t), \label{eq:20} \\
\mathrm{Re} \Bigl [
\frac{1}{\pi} \int_{-\ell}^\ell \frac{g'(\tau)}{\tau - t}d\tau& - \frac{\kappa-1}{\pi i (\kappa+1)} \int_{-\ell}^\ell \frac{q(\tau)}{\tau - t}d\tau + 2 \mathrm{Re}\, \Gamma + \bar \Gamma' 
\Bigr ] \\
&= -\frac{1}{4\mu}\sigma_0(\kappa+1) \mathrm{Re}\, g''(t) + \frac{1}{4\mu}(\zeta + 2\eta)(\kappa+1) \mathrm{Re}\, g''''(t), \label{eq:21} \\
-\mathrm{Im}\, q(t) &= \frac{\lambda_s + 2\mu_s}{4\mu}
\Bigl [-\frac{d}{dt} + \ell_s^2 \frac{d^3}{dt^3} \Bigr ] \mathrm{Re} \Bigl [
\frac{\kappa-1}{\pi} \int_{-\ell}^{ \ell} \frac{g'(\tau)}{\tau -t}d\tau + \frac{4\kappa}{\pi i (\kappa+1)} \int_{-\ell}^\ell \frac{q(\tau)}{\tau - t}d\tau 
\Bigr ], \label{eq:22}\\
\mathrm{Re}\, q(t) &=
\Bigl [-\frac{\sigma_0}{4\mu}\frac{d}{dt} + \frac{\zeta + 2 \eta}{4\mu}\frac{d^3}{dt^3} \Bigr ] \mathrm{Im} \Bigl [
\frac{\kappa-1}{\pi} \int_{-\ell}^{ \ell} \frac{g'(\tau)}{\tau -t}d\tau + \frac{4\kappa}{\pi i (\kappa+1)} \int_{-\ell}^\ell \frac{q(\tau)}{\tau - t}d\tau 
\Bigr ] \label{eq:23}. 
\end{align}

We now rescale \eqref{eq:20}, \eqref{eq:21}, \eqref{eq:22}, and \eqref{eq:23}. We introduce 
\begin{align}
	Q(t) = \ell q(\ell t), \quad G(t) = g(\ell t). 
\end{align}
Then \eqref{eq:20}-\eqref{eq:23} are equivalent to the following system on $[-1,1]$: 
\begin{align}
	-\frac{1}{\pi} \int_{-1}^1 \frac{\mathrm{Im}\,G'(\tau)}{\tau - t}d\tau& - \frac{\kappa-1}{\pi (\kappa+1)} \int_{-1}^1 \frac{\mathrm{Re}\, Q(\tau)}{\tau - t}d\tau + \ell \mathrm{Im}\, \Gamma' \\
	&= \gamma_1(\kappa+1) \mathrm{Im}\, G''(t) - \gamma_4(\kappa+1) \mathrm{Im}\, G''''(t), \label{eq:24} \\
	\frac{1}{\pi} \int_{-1}^1 \frac{\mathrm{Re}\, G'(\tau)}{\tau - t}d\tau& - \frac{\kappa-1}{\pi (\kappa+1)} \int_{-1}^1 \frac{\mathrm{Im}\,Q(\tau)}{\tau - t}d\tau + 2 \ell\mathrm{Re}\, \Gamma + \ell \mathrm{Re} \Gamma' \\
	&= -\gamma_2(\kappa+1) \mathrm{Re}\, G''(t) + \gamma_3(\kappa+1) \mathrm{Re}\, G''''(t), \label{eq:25} \\
	-\mathrm{Im}\, Q(t) &=
	\Bigl [-\gamma_1\frac{d}{dt} + \gamma_4 \frac{d^3}{dt^3} \Bigr ] \Bigl [
	\frac{\kappa-1}{\pi} \int_{-1}^{ 1} \frac{\mathrm{Re}\,G'(\tau)}{\tau -t}d\tau + \frac{4\kappa}{\pi (\kappa+1)} \int_{-1}^1 \frac{\mathrm{Im}\, Q(\tau)}{\tau - t}d\tau 
	\Bigr ], \label{eq:26}\\
	\mathrm{Re}\, Q(t) &=
	\Bigl [-\gamma_2 \frac{d}{dt} + \gamma_3 \frac{d^3}{dt^3} \Bigr ] \Bigl [
	\frac{\kappa-1}{\pi} \int_{-1}^{ 1} \frac{\mathrm{Im}\, G'(\tau)}{\tau -t}d\tau - \frac{4\kappa}{\pi (\kappa+1)} \int_{-1}^1 \frac{\mathrm{Re}\, Q(\tau)}{\tau - t}d\tau 
	\Bigr ] \label{eq:27}. 
\end{align}
where
\begin{align}
	\gamma_1 = \frac{2\mu_s + \lambda_s}{4\mu \ell} > 0, \quad \gamma_2 = \frac{\sigma_0}{4\mu \ell}, \quad \gamma_3 = \frac{\zeta + 2\eta}{4\mu \ell^3} > 0, \quad \gamma_4 = \frac{\ell_s^2(2\mu_s + \lambda_s)}{4 \mu \ell^3} > 0. \label{eq:28}
\end{align}
We require the tip conditions
\begin{align}
	G(\pm 1) = 0, \quad G'(\pm 1) = 0, \quad Q(\pm 1) = 0. \label{eq:29}
\end{align}
The first two above ensure the cusp conditions \eqref{eq:6} and \eqref{eq:7} with the final condition in \eqref{eq:29} enforcing continuity of the stresses at the tips; see \eqref{eq:8}, \eqref{eq:9}.

The system \eqref{eq:24}-\eqref{eq:27} can be separated into two decoupled systems which we analyze in the next section:
\begin{align}
	\frac{1}{\pi} \int_{-1}^1 \frac{\mathrm{Re}\, G'(\tau)}{\tau - t}d\tau& - \frac{\kappa-1}{\pi (\kappa+1)} \int_{-1}^1 \frac{\mathrm{Im}\,Q(\tau)}{\tau - t}d\tau + 2 \ell\mathrm{Re}\, \Gamma + \ell \mathrm{Re} \Gamma' \\
	&= -\gamma_2(\kappa+1) \mathrm{Re}\, G''(t) + \gamma_3(\kappa+1) \mathrm{Re}\, G''''(t), \label{eq:32} \\
	-\mathrm{Im}\, Q(t) &=
	\Bigl [-\gamma_1\frac{d}{dt} + \gamma_4 \frac{d^3}{dt^3} \Bigr ] \Bigl [
	\frac{\kappa-1}{\pi} \int_{-1}^{ 1} \frac{\mathrm{Re}\,G'(\tau)}{\tau -t}d\tau + \frac{4\kappa}{\pi (\kappa+1)} \int_{-1}^1 \frac{\mathrm{Im}\, Q(\tau)}{\tau - t}d\tau 
	\Bigr ]
	\label{eq:33},
\end{align}
and
\begin{align}
	-\frac{1}{\pi} \int_{-1}^1 \frac{\mathrm{Im}\,G'(\tau)}{\tau - t}d\tau& - \frac{\kappa-1}{\pi (\kappa+1)} \int_{-1}^1 \frac{\mathrm{Re}\, Q(\tau)}{\tau - t}d\tau + \ell \mathrm{Im}\, \Gamma' \\
&= \gamma_1(\kappa+1) \mathrm{Im}\, G''(t) - \gamma_4(\kappa+1) \mathrm{Im}\, G''''(t), \label{eq:30} \\
	\mathrm{Re}\, Q(t) &=
\Bigl [-\gamma_2 \frac{d}{dt} + \gamma_3 \frac{d^3}{dt^3} \Bigr ] \Bigl [
\frac{\kappa-1}{\pi} \int_{-1}^{ 1} \frac{\mathrm{Im}\, G'(\tau)}{\tau -t}d\tau - \frac{4\kappa}{\pi (\kappa+1)} \int_{-1}^1 \frac{\mathrm{Re}\, Q(\tau)}{\tau - t}d\tau 
\Bigr ]. \label{eq:31}
\end{align}

\section{Analysis of the governing equations}
In this section, we show that the system of singular integro-differential equations \eqref{eq:32},\eqref{eq:33} can be reformulated as a system of integral equations with Hilbert-Schmidt kernels. This reformulation establishes their unique solvability via the Fredholm alternative, except possibly for a countable set of values of the parameters $\gamma_1$, $\gamma_2$, $\gamma_3$, $\gamma_4$. We then prove sufficient regularity of solutions to guarantee the boundedness of the associated stress and strain fields up to the crack tips.  An analogous proof applies to \eqref{eq:30}-\eqref{eq:31} and will be omitted. 

\subsection{Change of variables}
We reformulate \eqref{eq:32}, \eqref{eq:33} in terms of two unknown functions $\chi_1, \chi_2$ satisfying 
\begin{align}
	\chi_1(t) &= \ReG''''(t), \label{eq:34} \\
	\chi_2(t) &= 
	2 \mu \ell^2 (u_{2,11}^+(\ell t) + u_{2,11}^-(\ell t)) \label{eq:38}\\
	&= \frac{d}{dt} \Bigl [
	\frac{\kappa-1}{\pi} \int_{-1}^1 \frac{\ReG'(\tau)}{\tau - t} d\tau + \frac{4\kappa}{\pi(\kappa+1)} \int_{-1}^1 \frac{\ImQ(\tau)}{\tau - t} d\tau. 
	\Bigr ] \label{eq:35} \\
	&= \frac{\kappa-1}{\pi} \int_{-1}^1 \frac{\ReG''(\tau)}{\tau - t} d\tau + \frac{4\kappa}{\pi(\kappa+1)} \int_{-1}^1 \frac{\ImQ'(\tau)}{\tau - t} d\tau. \label{eq:36}
\end{align}
The last equality follows from \eqref{eq:29}. By \eqref{eq:couple} and \eqref{eq:38}, we require that $\chi_2$ satisfy 
\begin{align}
	\chi_2(\pm 1) = 0. \label{eq:40}
\end{align}

We denote the following Green function 
\begin{align}
	\bbG(t,\tau) = 
	\begin{cases}
		\frac{1}{24}(t-1)^2(\tau+1)^2(1 + 2t - 2\tau - t\tau) &\quad \tau \in [-1,t], \\
		\frac{1}{24}(\tau-1)^2(t+1)^2(1 + 2\tau - 2t - t\tau) &\quad \tau \in [t,1], 
	\end{cases} 
\end{align}
satisfying 
\begin{gather}
\bbG_{tttt}(t,\tau) = \delta(t-\tau), \\
\bbG(\pm 1, \tau) = 0, \quad \bbG_t(\pm 1, \tau) = 0. \label{eq:44}
\end{gather}
Then $\bbG(t,\tau) = \bbG(\tau,t)$ for all $\tau,t \in [-1,1]$ and $\bbG \in C^{2,1}([-1,1] \times [-1,1])$. To satisfy \eqref{eq:34}, we set
\begin{align}
	\ReG(t) = \int_{-1}^1 \bbG(t,\tau) \chi_1(\tau) d\tau. \label{eq:37}
\end{align}
for an unknown function $\chi_1$. We note that \eqref{eq:37} automatically enforces half of the cusp conditions at the crack tips by \eqref{eq:44}:
\begin{align}
	\ReG(\pm 1) = 0, \quad \ReG'(\pm 1) = 0. 
\end{align}

We now invert \eqref{eq:36}. It is classical (see e.g., \cite{tricomi1951finite}) that 
\begin{align}
	f(x) = \frac{1}{\pi} \int_{-1}^{-1} \frac{\varphi(\tau)}{\tau - t} d\tau \implies \varphi(t) = \frac{C}{\sqrt{1-t^2}} - \frac{1}{\pi} \int_{-1}^1 \sqrt{\frac{1-\tau^2}{1-t^2}} \frac{f(\tau)}{\tau - t} d\tau
\end{align}
where $C = \frac{1}{\pi} \int_{-1}^1 \varphi(t)dt.$ Then \eqref{eq:36} implies that,
\begin{align}
	(\kappa-1) \ReG''(t) + \frac{4\kappa}{\kappa+1}\ImQ'(t) = \frac{C}{\sqrt{1-t^2}} - \frac{1}{\pi} \int_{-1}^1 \sqrt{\frac{1-\tau^2}{1-t^2}} \frac{\chi_2(\tau)}{\tau - t} d\tau, 
\end{align}
with 
\begin{align}
	\pi C = \int_{-1}^1 \Bigl (
	(\kappa-1) \ReG''(t) + \frac{4\kappa}{\kappa+1}\ImQ'(t)
	\Bigr ) dt = \left[(\kappa-1) \ReG'(t) + \frac{4\kappa}{\kappa+1}\ImQ(t)\right] \Big |_{-1}^1 = 0,
\end{align}
by \eqref{eq:29}. 
Thus, to satisfy \eqref{eq:36}, we must have 
\begin{align}
	(\kappa-1) \ReG''(t) + \frac{4\kappa}{\kappa+1}\ImQ'(t) &= - \frac{1}{\pi} \int_{-1}^1 \sqrt{\frac{1-\tau^2}{1-t^2}} \frac{\chi_2(\tau)}{\tau - t} d\tau.
\end{align}

Integrating from $-1$ to $t$ and using \eqref{eq:29} implies 
\begin{align}
	(\kappa-1) \ReG'(t) + \frac{4\kappa}{\kappa+1}\ImQ(t) &= -\frac{1}{\pi} \int_{-1}^t \int_{-1}^1 \sqrt{\frac{1-\tau^2}{1-s^2}} \frac{\chi_2(\tau)}{\tau - s} d\tau ds, \\
	&= \int_{-1}^1 \omega_1(t,\tau) \chi_2(\tau)d\tau  	\label{eq:39}
\end{align}
where 
\begin{gather}
	\omega_1(t,\tau) := -\frac{1}{\pi}\sqrt{1-\tau^2} \int_{-1}^t \frac{1}{\sqrt{1-s^2}}\frac{1}{\tau -s}ds \\
	= -\frac{1}{\pi}\log \Bigl |
	\frac{[\tau(1 - \sqrt{1-t^2}) - t(1 + \sqrt{1-\tau^2})][\tau + 1 - \sqrt{1-\tau^2}]}{[\tau(1-\sqrt{1-t^2}) - t (1 - \sqrt{1 - \tau^2})][\tau + 1 + \sqrt{1-\tau^2}]}
	\Bigr |,
\end{gather}
and thus, we have 
\begin{align}
 \ImQ(t) = \frac{\kappa+1}{4 \kappa} \int_{-1}^1 \omega_1(t,\tau) \chi_2(\tau)d\tau + \frac{1-\kappa^2}{4 \kappa} \int_{-1}^1 \bbG_{t}(t,\tau) \chi_1(\tau) d\tau. \label{eq:41} 
\end{align}
Moreover, since $\bbG'(\pm 1,\cdot) \equiv 0$ and 
\begin{align}
\forall |\tau| \neq 1, \quad \omega_1(\pm 1, \tau) = 0, 
\end{align}
\eqref{eq:41} implies that half of the continuity of stress conditions hold at the crack tips: 
\begin{align}
\ImQ(\pm 1) = 0. 
\end{align}

\subsection{Regularization of the system}
We now express \eqref{eq:32} and \eqref{eq:33} as integral equations using \eqref{eq:34}, \eqref{eq:38}, \eqref{eq:36}, \eqref{eq:40}, \eqref{eq:37}, and \eqref{eq:41}. We denote the following Green function 
\begin{align}
	\bbK = \begin{cases}
		\frac{1}{2}(t-1)(\tau+1) \quad &-1 \leq \tau \leq t \leq 1, \\
		\frac{1}{2}(t+1)(\tau-1) \quad &-1 \leq t \leq \tau \leq 1,
	\end{cases}
\end{align}   
satisfying 
\begin{align}
\bbK_{tt}(t,\tau) = \delta(t-\tau), \quad \bbK(\pm 1, \tau) = 0. 
\end{align}
Then \eqref{eq:33}, \eqref{eq:35}, \eqref{eq:36} and \eqref{eq:40} imply that 
\begin{align}
\begin{split}
	\gamma_4 \chi_2(t) 
	&= \gamma_1 \int_{-1}^1 \bbK(t,\tau) \chi_2(\tau) d\tau - \frac{\kappa+1}{4 \kappa} \int_{-1}^1 \int_{-1}^1  \bbK(t,\tau) \omega_1(\tau,s) \chi_2(s) ds d\tau 
    \\&\quad -\frac{1-\kappa^2}{4 \kappa} \int_{-1}^1 \int_{-1}^1 \bbK(t,\tau) \bbG_{\tau}(\tau,s) \chi_1(s) ds d\tau.
\end{split}\label{eq:42}
\end{align}
We also conclude from \eqref{eq:41} and \eqref{eq:32} that  
\begin{align}
\begin{split}
\gamma_3(\kappa+1) \chi_1(t) &= \gamma_2(\kappa+1)\int_{-1}^1 \bbG_{tt}(t,\tau) \chi_1(\tau) d\tau + \frac{(\kappa+1)^2}{4\pi \kappa} \int_{-1}^1 \int_{-1}^1 \frac{\bbG_\tau(\tau,s)}{\tau - t} \chi_1(s) ds d\tau \\ &\quad- \frac{\kappa-1}{4\pi \kappa} \int_{-1}^1 \int_{-1}^1 \frac{\omega_1(\tau,s)}{\tau -t} \chi_2(s) ds d\tau + 
2 \ell\mathrm{Re}\, \Gamma + \ell \mathrm{Re} \Gamma'.
\end{split}\label{eq:43}
\end{align}
In summary, we seek $\chi_1, \chi_2 \in L^2(-1,1)$ that satisfies the pair of integral equations \eqref{eq:43} and \eqref{eq:42}. 

We observe that each integral operator on the right-hand sides of \eqref{eq:42} and \eqref{eq:43} is a compact operator on $L^2(-1,1)$. Indeed, consider \eqref{eq:42}. The first and third integral operators on the right-hand side of \eqref{eq:42} have continuous kernels in $(t,s)$ due to the smoothness properties of $\bbG$ and $\bbK$, yielding compact integral operators on $L^2(-1,1)$. Since $\omega_1 \in L^2((-1,1) \times (-1,1))$ is a Hilbert-Schmidt kernel, the second integral operator is the composition of two compact operators, $f \mapsto \int_{-1}^1 \bbK(\cdot,\tau) g(\tau) d\tau$ and $g \mapsto \int_{-1}^1 \omega(\cdot, s) g(s) ds$, and therefore it is also a compact operator. 

Next, we consider the right-hand side of \eqref{eq:43}. Again, the smoothness properties of $\bbG$ imply that the first integral operator has a continuous kernel. The second integral operator is the composition of two integral operators, one bounded on $L^2(-1,1)$ given by the Hilbert transform, and the other a compact operator given by $f \mapsto \int_{-1}^1 \bbG_\tau(\tau,s) f(s)ds$. Since the composition of a bounded operator and a compact operator is a compact operator, it follows that the second integral operator on the right-hand side of \eqref{eq:43} is a compact operator on $L^2(-1,1)$. Similarly, the third integral operator is a compact operator on $L^2(-1,1)$. 

In conclusion, all integral operators appearing in \eqref{eq:42} and \eqref{eq:43} are compact operators on $L^2(-1,1)$. By the Fredholm alternative, \eqref{eq:42} and \eqref{eq:43} are uniquely solvable for $\chi_1, \chi_2 \in L^2(-1,1)$ for all but a countable set of values of the parameters $\gamma_1$, $\gamma_2$, $\gamma_3$, $\gamma_4$.  

\subsection{Regularity of the solutions} In this section, we demonstrate that the solutions obtained in the previous section produce stresses and strains that possess finite limits at the crack tips. 

We first note that by \eqref{eq:37}, it follows that $\mathrm{Re}\, G(t) \in H^4(-1,1)$. Sobolev embedding implies that $\mathrm{Re}\, G(t) \in C^{3,1/2}([-1,1])$, and again \eqref{eq:44} implies that half of the cusp conditions at the crack tips are satisfied: 
\begin{align}
    \mathrm{Re}\, G(\pm 1) = 0, \quad \mathrm{Re}\, G'(\pm 1) = 0.
\end{align}
Similar arguments from the previous two sections applied to \eqref{eq:30}-\eqref{eq:31} imply that $\mathrm{Im}\, G(t) \in C^{3,1/2}([-1,1])$ and the other half of the cusp conditions at the crack tips are satisfied
\begin{align}
\mathrm{Im}\, G(\pm 1) = 0, \quad \mathrm{Im}\, G'(\pm 1) = 0. 
\end{align}

To analyze the regularity of $Q$, we consider \eqref{eq:41}, which we rewrite as 
\begin{align}
    \mathrm{Im}\, Q(t) = \frac{\kappa+1}{4\kappa} \int_{-1}^t \int_{-1}^1 \frac{\sqrt{1-\tau^2}}{\sqrt{1-s^2}}\frac{\chi_2(\tau)}{\tau-s}d\tau ds + \frac{1-\kappa^2}{4\kappa} \int_{-1}^1 \bbG_t(t,\tau) \chi_1(\tau)d\tau. \label{eq:55}
\end{align}
The second term is clearly continuously differentiable since $\bbG \in C^2([-1,1] \times [-1,1])$. By Tricomi's boundedness result (p. 204, \cite{tricomi1951finite}),  $\chi_2 \in L^2(-1,1)$ implies that 
\begin{align}
\phi(s) := \int_{-1}^1 \frac{\sqrt{1-\tau^2}}{\sqrt{1-s^2}}\frac{\chi_2(\tau)}{\tau-s}d\tau \in L^q(-1,1), \quad q \in (1, 4/3). 
\end{align}
Thus, $\mathrm{Im}\, Q \in W^{1,q}(-1,1)$ for every $q \in (1,4/3)$. By Sobolev embedding and the fact that 
\begin{align}
\int_{-1}^1 \int_{-1}^1 \frac{\sqrt{1-\tau^2}}{\sqrt{1-s^2}}\frac{\chi_2(\tau)}{\tau-s}d\tau ds = 0, 
\end{align}
we conclude that for every $r \in (0, 1/4)$, 
\begin{align}
 \mathrm{Im}\, Q \in C^r([-1,1]) \quad \mbox{ and } \quad \mathrm{Im}\, Q(\pm 1) = 0. 
\end{align}
Similar arguments from the previous two sections applied to \eqref{eq:30}-\eqref{eq:31} imply that $\mathrm{Re}\, Q(t) \in C^{r}([-1,1])$ for all $r \in (0,1/4)$ and the other half of the continuity of stress conditions hold at the crack tips:
\begin{align}
\mathrm{Re}\, Q(\pm 1) = 0. 
\end{align}

In summary, and after rescaling, we conclude that
\begin{gather}
g \in C^{3,1/2}([-\ell,\ell]), \quad g(\pm \ell) = g'(\pm \ell) = 0, \label{eq:46} \\
\forall r \in (0,1/4), \quad q \in C^r([-\ell,\ell]), \quad q(\pm \ell) = 0.  \label{eq:47}
\end{gather}
The finite Hilbert transform on $[-\ell, \ell]$ of a H\"older continuous function $f : [-\ell,\ell] \rar \bbC$ satisfying $f(\pm \ell) = 0$ yields a bounded H\"older continuous function on $[-\ell,\ell]$ via extending $f$ to 0 outside of $[-\ell, \ell]$ and using the classical result for singular integral operators acting on Holder spaces (see Chapter 7 of \cite{MuscaluSchlag2013VolI}). This fact, \eqref{eq:46}, \eqref{eq:47}, \eqref{eq:14}, and \eqref{eq:15} imply that the stresses and strains have finite limits at the crack tips.

\section{Numerical solution of the fracture problem}

This section presents numerical solutions of the mixed-mode problem based on Chebyshev polynomial approximations of the integral equation system. The computations illustrate how the surface strain-gradient moduli influence the near-tip fields and confirm that the stresses and strains remain bounded for a broad range of admissible parameter values. These results complement the analytical theory and quantify the magnitude of the regularization produced by the surface strain-gradient terms.

\subsection{Numerical method} Consider the systems \eqref{eq:32}, \eqref{eq:33} and \eqref{eq:30}, \eqref{eq:31}. In order to find the numerical solution to these systems we will start with the Chebyshev polynomial approximations of the unknown functions in the form:
\begin{equation}
    \ImG'(t)=\sum_{k=1}^{N+2}A_kU_k(t)\sqrt{1-t^2},\,\,\,\ReQ(t)=\sum_{k=0}^{N+1}B_kU_k(t)\sqrt{1-t^2},
    \label{eq:Num1}
\end{equation}
\begin{equation}
    \ReG'(t)=\sum_{k=1}^{N+2}C_kU_k(t)\sqrt{1-t^2},\,\,\,\ImQ(t)=\sum_{k=0}^{N+1}D_kU_k(t)\sqrt{1-t^2},
    \label{eq:Num2}
\end{equation}
where $U_k(t)$ denotes a Chebyshev polynomial of the second kind of the degree $k$, and $N$ corresponds to the desired polynomial degree of approximation for the unknown functions. Note that taking the approximations of the unknown functions in the forms \eqref{eq:Num1}, \eqref{eq:Num2} allows one to automatically satisfy the tip conditions \eqref{eq:29}. The last two conditions \eqref{eq:29} are satisfied due to the square root term in the formulas \eqref{eq:Num1}, \eqref{eq:Num2}, while the first conditions \eqref{eq:29} are satisfied due to the absence of the first term $k=0$ in the formulas for $\ReG'(t)$ and $\ImG'(t)$ and the fact that we can integrate the function $G'(t)$. 

Taking the corresponding derivatives of the formulas \eqref{eq:Num1}, \eqref{eq:Num2} leads to the expressions:
\begin{equation}
    \ImG''(t)=-\sum_{k=1}^{N+2}(k+1)A_k\frac{T_{k+1}(t)}{\sqrt{1-t^2}},\,\,\,\ReG''(t)=-\sum_{k=1}^{N+2}(k+1)C_k\frac{T_{k+1}(t)}{\sqrt{1-t^2}},
    \label{eq:Num3}
\end{equation}
\begin{equation}
\ImG''''(t)=\sum_{k=1}^{N+2}(k+1)(k+2)(k+3)A_k\frac{T_{k+1}(t)}{(\sqrt{1-t^2})^3}-3\sum_{k=1}^{N+2}(k+1)(k+3)A_k\frac{U_{k+1}(t)}{(\sqrt{1-t^2})^3}-
    \label{eq:Num4}
\end{equation}
$$
3\sum_{k=1}^{N+2}(k+1)A_k\frac{T_{k+3}(t)}{(\sqrt{1-t^2})^5},
$$
\begin{equation}
\ReG''''(t)=\sum_{k=1}^{N+2}(k+1)(k+2)(k+3)C_k\frac{T_{k+1}(t)}{(\sqrt{1-t^2})^3}-3\sum_{k=1}^{N+2}(k+1)(k+3)C_k\frac{U_{k+1}(t)}{(\sqrt{1-t^2})^3}-
    \label{eq:Num5}
\end{equation}
$$
3\sum_{k=1}^{N+2}(k+1)C_k\frac{T_{k+3}(t)}{(\sqrt{1-t^2})^5},
$$
where $T_k(t)$ denotes a Chebyshev polynomial of the first kind of the degree $k$.

Moreover, utilizing the following formula for the Cauchy integral of the Chebyshev polynomials of the second kind:
\begin{equation}
\frac{1}{\pi}\int_{-1}^1\frac{U_k(t)\sqrt{1-t^2}dt}{t-x}=-T_{k+1}(x),
    \label{eq:Num6}
\end{equation}
we can obtain the following expressions for the singular integrals:
\begin{equation}
-\frac{1}{\pi} \int_{-1}^1 \frac{\ImG'(\tau)}{\tau - t} d\tau - \frac{\kappa-1}{\pi(\kappa+1)} \int_{-1}^1 \frac{\ReQ(\tau)}{\tau - t} d\tau=\sum_{k=0}^{N+2}\left(A_k+\frac{\kappa-1}{\kappa+1}B_k\right)T_{k+1}(t),
    \label{eq:Num7}
\end{equation}
\begin{equation}
\frac{\kappa-1}{\pi} \int_{-1}^1 \frac{\ImG'(\tau)}{\tau - t} d\tau - \frac{4\kappa}{\pi(\kappa+1)} \int_{-1}^1 \frac{\ReQ(\tau)}{\tau - t} d\tau=\sum_{k=0}^{N+2}\left(-(\kappa-1)A_k+\frac{4\kappa}{\kappa+1}B_k\right)T_{k+1}(t),
    \label{eq:Num8}
\end{equation}
\begin{equation}
\frac{1}{\pi} \int_{-1}^1 \frac{\ReG'(\tau)}{\tau - t} d\tau - \frac{\kappa-1}{\pi(\kappa+1)} \int_{-1}^1 \frac{\ImQ(\tau)}{\tau - t} d\tau=\sum_{k=0}^{N+2}\left(-C_k+\frac{\kappa-1}{\kappa+1}D_k\right)T_{k+1}(t),
    \label{eq:Num9}
\end{equation}
\begin{equation}
\frac{\kappa-1}{\pi} \int_{-1}^1 \frac{\ReG'(\tau)}{\tau - t} d\tau + \frac{4\kappa}{\pi(\kappa+1)} \int_{-1}^1 \frac{\ImQ(\tau)}{\tau - t} d\tau=\sum_{k=0}^{N+2}\left(-(\kappa-1)C_k-\frac{4\kappa}{\kappa+1}D_k\right)T_{k+1}(t),
    \label{eq:Num10}
\end{equation}
where, for the sake of brevity, it is assumed that 
\begin{equation}
A_0=0,\,\,\,C_0=0,\,\,\,B_{N+2}=0,\,\,\,D_{N+2}=0.
    \label{eq:Num11}
\end{equation}

Taking the derivatives and utilizing the properties of the Chebyshev polynomials leads to the following relations:
\begin{equation}
T'_{k+1}(t)=(k+1)U_k(t),
    \label{eq:Num12}
\end{equation}
\begin{equation}
T'''_{k+1}(t)=-(k+1)(k+2)(k+3)\frac{U_k(t)}{1-t^2}-3(k+1)(k+3)\frac{T_{k+2}(t)}{(1-t^2)^2}+3(k+1)\frac{U_{k+2}(t)}{(1-t^2)^2}.
    \label{eq:Num13}
\end{equation}

Substituting the formulas \eqref{eq:Num3}-\eqref{eq:Num5}, \eqref{eq:Num7}-\eqref{eq:Num10}, \eqref{eq:Num12}, \eqref{eq:Num13} into the systems of singular integral equations \eqref{eq:32}, \eqref{eq:33} and \eqref{eq:30}, \eqref{eq:31}, and utilizing the orthogonality property of the Chebyshev polynomials:
$$
\int_{-1}^1\frac{T_j(t)T_k(t)dt}{\sqrt{1-t^2}}=\left\{
\begin{array}{cc}
0,& j\neq k,\\
\pi/2, & j=k\neq 0,\\
\pi, & j=k=0,
\end{array}
\right.
$$
leads to the following systems of linear algebraic equations for the coefficients $A_k$, $B_k$, $C_k$, $D_k$:
$$
\sum_{k=0}^{N+2}\left(A_k+\frac{\kappa-1}{\kappa+1}B_k \right)I_{jk}^1+\mathrm{Im}\Gamma'\ell J_j=-\gamma_1(\kappa+1)\sum_{k=0}^{N+2}A_k(k+1)I_{jk}^2-
$$
$$
\gamma_4(\kappa+1)\sum_{k=0}^{N+2}A_k(k+1)(k+2)(k+3)I_{jk}^3+3\gamma_4(\kappa+1)\sum_{k=0}^{N+2}A_k(k+1)(k+3)I_{jk}^4+
$$
\begin{equation}
  3\gamma_4(\kappa+1)\sum_{k=0}^{N+2}A_k(k+1)\delta_{j,k+3}\frac{\pi}{2},\,\,\,j=0,\ldots,N,
  \label{eq:Num14}
\end{equation}
$$
\sum_{k=0}^{N+1}B_kI_{jk}^5=\gamma_2\sum_{k=0}^{N+2}\left((\kappa-1)A_k-\frac{4\kappa}{\kappa+1}B_k \right)(k+1)I_{jk}^6-
$$
$$
\gamma_3\sum_{k=0}^{N+2}\left((\kappa-1)A_k-\frac{4\kappa}{\kappa+1}B_k \right)\left(-(k+1)(k+2)(k+3)I_{jk}^7\right.
$$
\begin{equation}
\left.-3(k+1)(k+3)\delta_{j,k+2}\frac{\pi}{2}+3(k+1)I_{jk}^8 \right),\,\,\,j=0,\ldots,N,
\label{eq:Num15}
\end{equation}
$$
\sum_{k=0}^{N+2}\left(-C_k+\frac{\kappa-1}{\kappa+1}D_k \right)I_{jk}^1+(2\mathrm{Re}\Gamma+\mathrm{Re}\Gamma')\ell J_j=\gamma_2(\kappa+1)\sum_{k=0}^{N+2}C_k(k+1)I_{jk}^2+
$$
$$
\gamma_3(\kappa+1)\sum_{k=0}^{N+2}C_k(k+1)(k+2)(k+3)I_{jk}^3-3\gamma_3(\kappa+1)\sum_{k=0}^{N+2}A_k(k+1)(k+3)I_{jk}^4-
$$
\begin{equation}
  3\gamma_3(\kappa+1)\sum_{k=0}^{N+2}C_k(k+1)\delta_{j,k+3}\frac{\pi}{2},\,\,\,j=0,\ldots,N,
  \label{eq:Num16}
\end{equation}
$$
-\sum_{k=0}^{N+1}D_kI_{jk}^5=\gamma_1\sum_{k=0}^{N+2}\left((\kappa-1)C_k+\frac{4\kappa}{\kappa+1}D_k \right)(k+1)I_{jk}^6-
$$
$$
\gamma_4\sum_{k=0}^{N+2}\left((\kappa-1)C_k+\frac{4\kappa}{\kappa+1}D_k \right)\left(-(k+1)(k+2)(k+3)I_{jk}^7\right.
$$
\begin{equation}
\left.-3(k+1)(k+3)\delta_{j,k+2}\frac{\pi}{2}+3(k+1)I_{jk}^8 \right),\,\,\,j=0,\ldots,N,
\label{eq:Num17}
\end{equation}
where the coefficients $I_{jk}^l$ are given by the formulas:
$$
I_{jk}^1=\int_{-1}^1T_{k+1}(t)T_j(t)(1-t^2)^2dt,
$$
$$
I_{jk}^2=\int_{-1}^1T_{k+1}(t)T_j(t)(1-t^2)^{3/2}dt,
$$
$$
I_{jk}^3=\int_{-1}^1T_{k+1}(t)T_j(t)(1-t^2)^{1/2}dt,
$$
$$
I_{jk}^4=\int_{-1}^1U_{k+1}(t)T_j(t)(1-t^2)^{1/2}dt,
$$
\begin{equation}
I_{jk}^5=\int_{-1}^1U_{k}(t)T_j(t)(1-t^2)^2dt,   
\label{eq:Num18}
\end{equation}
$$
I_{jk}^6=\int_{-1}^1U_{k}(t)T_j(t)(1-t^2)^{3/2}dt,
$$
$$
I_{jk}^7=\int_{-1}^1U_{k}(t)T_j(t)(1-t^2)^{1/2}dt=I_{j(k-1)}^4,
$$
$$
I_{jk}^8=\int_{-1}^1\frac{U_{k+2}(t)T_j(t)}{\sqrt{1-t^2}}dt,
$$
$$
J_j^1=\int_{-1}^1T_j(t)(1-t^2)^2dt.
$$
The integrals in the formulas \eqref{eq:Num18} can be computed approximately by using Gauss-Chebyshev quadratures.

Finally, in order to close the systems of linear algebraic equations  \eqref{eq:Num14}, \eqref{eq:Num15} and \eqref{eq:Num16}, \eqref{eq:Num17}, we need to satisfy the tip conditions \eqref{eq:couple} and \eqref{eq:40}. These conditions lead to the equations:
\begin{equation}
\sum_{k=0}^{N+2}\left((\kappa-1)A_k-\frac{4\kappa}{\kappa+1}B_k \right)(k+1)^2(\pm 1)^k=0,
    \label{eq:Num19}
\end{equation}
\begin{equation}
\sum_{k=0}^{N+2}\left((\kappa-1)C_k+\frac{4\kappa}{\kappa+1}D_k \right)(k+1)^2(\pm 1)^k=0.
    \label{eq:Num20}
\end{equation}

The systems of linear algebraic equations \eqref{eq:Num14}, \eqref{eq:Num15}, \eqref{eq:Num19} and \eqref{eq:Num16}, \eqref{eq:Num17}, \eqref{eq:Num20} can then be solved to find the values of the coefficients $A_k$, $B_k$, $C_k$, $D_k$. These coefficients can be used in the formulas \eqref{eq:Num1}, \eqref{eq:Num2} to determine the unknown functions $\ReG'(t)$, $\ImG'(t)$, $\ReQ(t)$, $\ImQ(t)$, which can be further used to find the values of the stresses and the displacements through the formulas \eqref{eq:12}, \eqref{eq:13}. 

\subsection{Examples} The numerical examples have been computed for the following values of the mechanical and geometric parameters: $\mu=26.3\; \mbox{GPa}$, $\nu=0.2481$, $\gamma_1=\gamma_2=1.602\;\gamma\; \mbox{N/m}$, $\gamma_4=1.602\;\cdot 10^{-19}\gamma\; \mbox{N/m}^2$, where $\gamma=5$; $\sigma_{11}^{\infty}=\sigma_{22}^{\infty}=\sigma_{12}^{\infty}=0.5\; \mbox{GPa}$, and $\ell=5\; \mbox{nm}$. The new parameter $\gamma_4$ introduced in this paper takes on values $0.1\gamma_3$, $\gamma_3$, and $10\gamma_3$, and is equal to zero for the Steigmann-Ogden model. The results are plotted for $N=30$, and $100$ quadrature nodes in the computation of the integrals. Further increasing these values does not improve the accuracy of the computations.

\begin{figure}
\begin{center}
\includegraphics[width=0.45\textwidth]{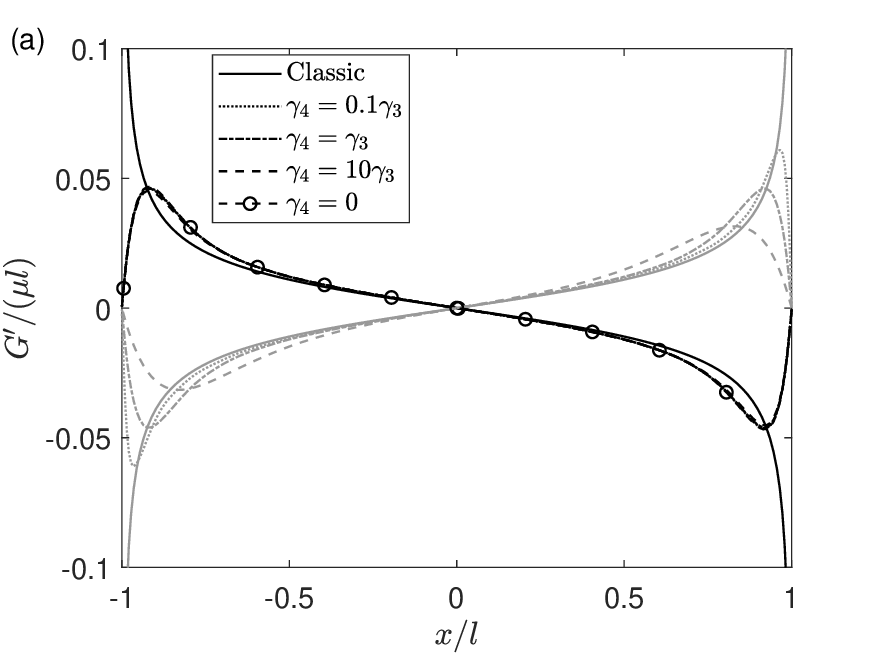} 
\includegraphics[width=0.45\textwidth]{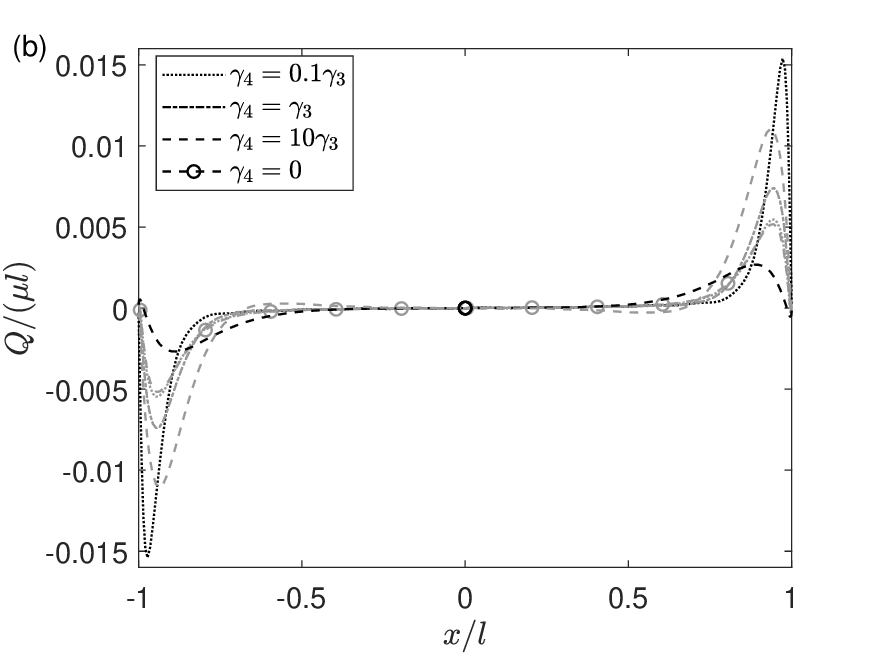} 
\end{center}
\caption{The values of the functions $G'(t)$ and $Q(t)$ for various values of the parameter $\gamma_4$.}
\label{fig2}
\end{figure} 

On Figure \ref{fig2}, the graphs of the functions $\mbox{Re}G'(t)$ (black lines), $\mbox{Im}G'(t)$ (gray lines), $\mbox{Re}Q(t)$ (black lines), and $\mbox{Im}Q(t)$ (gray lines) are shown for the following values of the parameter $\gamma_4$: $\gamma_4=0.1\gamma_3$ (dot lines), $\gamma_4=\gamma_3$ (dash-dot lines), $\gamma_4=10\gamma_3$ (dash lines). The results of the present study are compared to the classic results without taking into account any of the surface parameters which are shown by solid lines. The results have also been compared to the results for the Steigmann-Ogden model ($\gamma_4=0$) which are shown by dash lines with circular markers and have been previously reported in \cite{ZemMac2020}. It is observed that the results of the current study correspond well to the known in the literature results. It can be seen that incorporating additional surface effects represented by the parameter $\gamma_4$ can significantly affect the results. In particular, smoothing effect near the tips of the crack can be seen on Figure \ref{fig2}.  

\begin{figure}
\begin{center}
\includegraphics[width=0.45\textwidth]{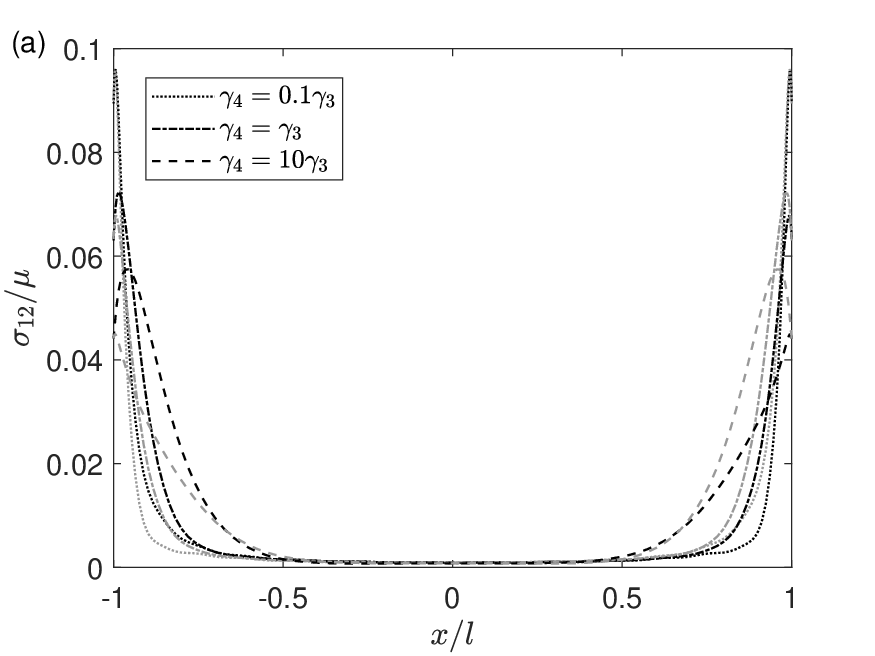} 
\includegraphics[width=0.45\textwidth]{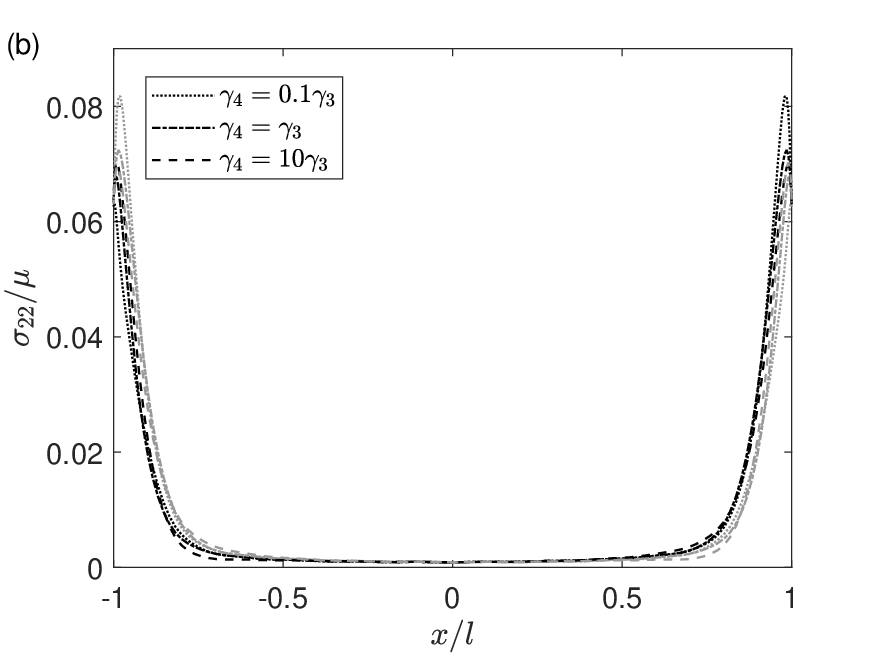} 
\end{center}
\caption{The values of the stresses $\sigma_{12}/\mu$ and $\sigma_{22}/\mu$ for various values of the parameter $\gamma_4$.}
\label{fig3}
\end{figure} 

The values of the normalized stresses $\sigma_{12}/\mu$ and $\sigma_{22}/\mu$ for the same values of the mechanical and geometric parameters as before are presented on Figure \ref{fig3}. The graphs of the stresses $\sigma_{12}^+/\mu$ (black lines) and $\sigma_{12}^-/\mu$ (gray lines) on the upper and lower bank of the crack are shown on Figure \ref{fig3}a, and  $\sigma_{22}^+/\mu$ (black lines) and $\sigma_{22}^-/\mu$ (gray lines) are shown on Figure \ref{fig3}b for the following values of the parameter $\gamma_4$: $\gamma_4=0.1\gamma_3$ (dot lines), $\gamma_4=\gamma_3$ (dash-dot lines), $\gamma_4=10\gamma_3$ (dash lines).  Again, it can be observed that the parameter $\gamma_4$ has a significant influence on the results and produces lower values of the stresses at the tips of the crack.

\begin{figure}
\begin{center}
\includegraphics[width=0.45\textwidth]{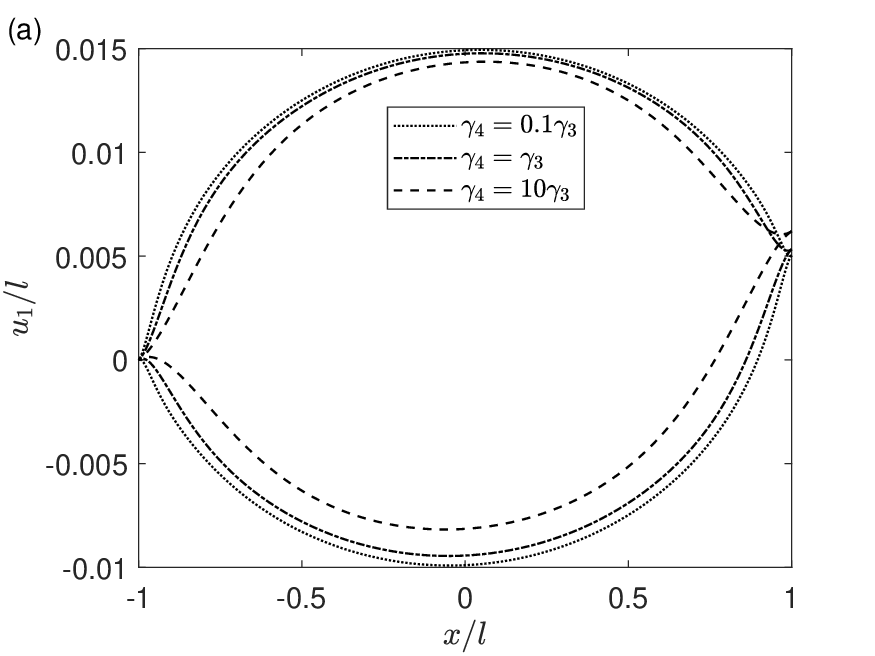} 
\includegraphics[width=0.45\textwidth]{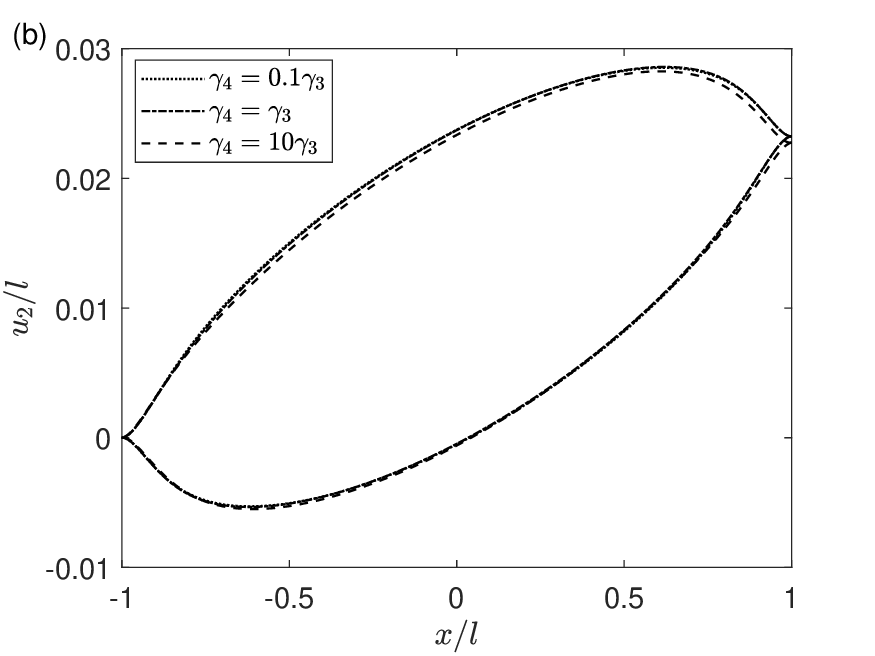} 
\end{center}
\caption{The values of the stresses $u_1/\ell$ and $u_2/\ell$ for various values of the parameter $\gamma_4$.}
\label{fig4}
\end{figure} 

Finally, the displacements $u_1/\ell$ and $u_2/\ell$ of the banks of the crack are shown on Figure \ref{fig4}. The graphs of the displacements $u_1^+/\ell$ and $u_1^-/\ell$ on the upper and lower bank of the crack are shown on Figure \ref{fig4}a, and  $u_2^+/\ell$ and $u_2^-/\ell$ are shown on Figure \ref{fig4}b for the following values of the parameter $\gamma_4$: $\gamma_4=0.1\gamma_3$ (dot lines), $\gamma_4=\gamma_3$ (dash-dot lines), $\gamma_4=10\gamma_3$ (dash lines). It can be seen that larger values $\gamma_4$ produce more cusp-like profiles near the tips of the crack and result in a smaller opening and less stretch of the crack. It can also be seen that the parameter $\gamma_4$ has larger influence on the displacement $u_1$ than $u_2$.

\section{Conclusions}
This paper studies a plane mixed mode fracture problem within a surface strain gradient elasticity framework previously introduced in \cite{rodriguez2024elastic, rodriguez2024midsurface} for mode III loading. The theory generalizes the Steigmann–Ogden surface energy by incorporating surface gradients of stretching in a variationally consistent manner.

The surface energy and governing equations are derived from a principle of virtual work, and the admissible crack tip conditions are analyzed in detail. These tip conditions reduce to crack closure in the sense of a cusp condition, continuity of stresses at the crack tips, and the absence of additional twisting moments at the crack edge. Using complex analytic representations for stresses and displacements, the fracture problem is reduced to a system of singular integro differential equations that can be regularized. The central analytical result is a rigorous proof that the resulting stress and strain fields remain bounded at the crack tips, a behavior that departs fundamentally from classical fracture mechanics.

A numerical scheme, employing Chebyshev polynomial representations, is derived and the resulting solutions are compared with results from the classical Steigmann-Ogden model \cite{ZemMac2020}. The additional surface strain-gradient term is observed to produce smoother crack tip closure and a smaller crack opening displacement. Taken together with the results of \cite{rodriguez2024elastic, rodriguez2024midsurface}, the analytical and numerical findings demonstrate that gradient strain surface elasticity yields bounded stress and strain fields for all modes of plane brittle fracture.

An important direction for future work is the formulation of fracture criteria based on finite crack tip stresses or strains within this framework. In particular, the absence of singularities in all modes of loading suggests the possibility of defining fracture criteria directly in terms of critical stress \emph{values} at the crack tips, rather than stress intensity factors. Such criteria could provide a physically transparent alternative to classical fracture measures and merit further analytical and numerical investigation.

\bibliographystyle{plain}
\bibliography{researchbibmech, researchbibmech_NSF}

\bigskip
\footnotesize

\noindent \textsc{Department of Mathematics, The University of North Carolina at Chapel Hill, CB 3250 Phillips Hall,
Chapel Hill, NC 27599}\\
\noindent \textit{E-mail address}: \texttt{crodrig@email.unc.edu}

\bigskip

\noindent \textsc{Department of Mathematics, Kansas State University, 138 Cardwell Hall,
Manhattan, Kansas, 66506}\\
\noindent \textit{E-mail address}: \texttt{azem@ksu.edu}

\end{document}